\documentclass{ifacconf}      
\usepackage{graphicx,amssymb,amsfonts,amsmath}
\usepackage{enumerate,latexsym,multicol,xcolor}
\usepackage{epstopdf,natbib,mathrsfs,subfig}
\DeclareMathAlphabet{\mathpzc}{OT1}{pzc}{m}{it}

\newtheorem{theorem}{Theorem}
\newtheorem{remark}{Remark}
\newtheorem{definition}{Definition}

\newtheorem{lemma}{Lemma}
\newtheorem{assumption}{Assumption}\newenvironment{proof}[1][Proof:]{\begin{trivlist}
\item[\hskip \labelsep {\it #1}]}{\end{trivlist}}

\newcommand{\rref}[1]{(\ref{#1})}
\begin{document}
\begin{frontmatter}

\title{On Output Feedback Stabilization of Time-Varying Decomposable Systems with Switching Topology and Delay$^\star$}\thanks{This paper was presented at the 21st IFAC World Congress; July 12-17; 2020; Berlin, Germany} 

\author[First]{Muhammad Zakwan} 
\author[Second]{Saeed Ahmed} 
\author[Second]{Naim Bajcinca} 

\address[First]{Department of Electrical and Electronics Engineering, Bilkent University, 06800 Ankara, Turkey (e-mail: zakwan@ ee.bilkent.edu.tr)}
\address[Second]{Department of Mechanical and Process Engineering, University of Kaiserslautern, 67663 Kaiserslautern, Germany\\ (e-mails: saeed.ahmed@mv.uni-kl.de, naim.bajcinca@mv.uni-kl.de)}

\begin{abstract} 
This paper presents a new method for dynamic output feedback stabilizing controller design for decomposable systems with switching topology and delay. Our approach consists of two steps. In the first step, we model the decomposable systems with switching topology as equivalent LPV systems with a piecewise constant parameter. In the second step, we design stabilizing output feedbacks for these LPV systems in the presence of a time-varying output delay using a  trajectory-based stability analysis approach. We do not impose any constraint on the delay derivative.  Finally, we illustrate our approach by applying it to the consensus problem of non-holonomic agents.
\end{abstract}


\end{frontmatter}

\section{Introduction}
A distributed system is a swarm of subsystems that are connected physically or through communication protocols with each subsystem having information about the interconnection topology.  Such systems emerge in many application domains such as vehicle platooning  (\citealp{jovanovic2005ill}), multi-UAV formation flight (\citealp{BeVePrTa}),  satellite formation  (\citealp{MeHa01,Ca00}), paper machine problem (\citealp{StGoDu03}), and large segmented telescopes (\citealp{JiVoHo06}). Motivated by these real-world applications, many researchers have studied various problems related to distributed systems such as consensus problem, flocking problem, and formation problem; see \cite{li2014cooperative}. The decomposable system (or identical dynamically decoupled
system), on the other hand, is a special class of distributed systems with identical subsystems interacting with each other. Tools from the algebraic graph
theory, such as Laplacian matrices, or graph-adjacency matrices (known as pattern matrices) are used to represent
the interactions among the subsystems of a decomposable system; see \cite{BoKe06}. A general overview of these systems, and their applications are provided in \cite{massioni2009distributed}, \cite{ghadami2013decomposition}, and \cite{eichler2014robust}. 

Since time delay may affect the performance of a distributed network in practice due to non-ideal signal transmission, the study of distributed systems with a delay is strongly motivated.  Therefore, many efforts have been made in the literature to tackle the issues of stability and performance degradation caused by communication or network delay in distributed systems;  see \cite{atay2013consensus}, \cite{ghaedsharaf2016interplay}, \cite{OlMu04}, \cite{papachristodoulou2010effects},  \cite{qiao2016consensus},  \cite{seuret2008consensus}, and \cite{SuWa09}. Apart from time delays, another interesting phenomenon in distributed systems is switching topology, where the interconnection
links may change over time due to various reasons. For example, communicating mobile agents may lose an existing connection due to the presence of an obstacle. On the other hand, a new connection may be established between the agents when they come close to each other in an effective range of detection.

In this paper, we provide a new method for dynamic output feedback stabilization of time-varying decomposable systems with switching topology and delay. Our technique involves two steps. First, we  model a decomposable system with switching topology as an equivalent LPV system with a piecewise constant parameter. Then, we use a trajectory-based approach to design output feedback controllers, ensuring the stability of this class of LPV systems in the presence of a time-varying pointwise output delay. Both of these steps are important for their own sakes and can be considered as two seperate contributions of this paper.  Motivated by the serious obstacle presented by the search for suitable Lyapunov functionals for switched and LPV systems with delay, we employ a trajectory-based stability result, proposed in  \cite{ahmed2018dynamic}.  We allow the delay to be a piecewise continuous function of time, and we do not impose any constraint on the delay derivative, which makes it possible to apply our approach to systems where the delay cannot be approximated by a differentiable delay with a bounded first derivative. Typical examples of this phenomenon include data flow across a communication network and delay resulting from sampling. While \cite{ahmed2018dynamic} considers switched systems with countable modes, here we extend their results  to switched systems with an uncountable number of modes. Stability analysis and control of LPV systems with piecewise constant parameters is also presented in \cite{briat2015stability}. However, there are two key differences between \cite{briat2015stability} and the present work, (i)  no delay is present in \cite{briat2015stability}, (ii) we study output feedback control, whereas state feedback control is discussed in \cite{briat2015stability}. Our work can be regarded as an extension of \cite{zakwan2019distributed}, offering new advantages, because (i) we use a 
trajectory-based approach for stability analysis which circumvents
the serious obstacle presented by the search for appropriate Lyapunov functionals, (ii) we do not impose any constraint
on the upper bound of the delay derivative.

The paper unfolds as follows. Section~2 presents modeling of decomposable systems with switching topology as equivalent LPV systems with a piecewise constant parameter. The output feedback stabilizing controller design appears in Section~3. The application of our results to multi-agent nonholonomic systems is presented in Section~4, and Section~5  presents concluding remarks and some future perspectives.

The notation will be simplified whenever no confusion can arise from the context. The identity matrix of appropriate dimension and the Kronecker product are denoted by $I$  and $\otimes$, respectively. The set of real numbers and the set of nonnegative real numbers are denoted by $\mathbb{R}$ and $\mathbb{R}_{\geq 0}$, respectively. The set of positive integers and the set of whole numbers are denoted by $\mathbb{N}$ and $\mathbb{N}_0 := \mathbb{N} \cup \{0\}$, respectively.  The usual Euclidean norm of vectors, and the induced norm of matrices, are denoted by $|\cdot|$.
Given any constant $\tau > 0$, we let $C([- \tau, 0], \mathbb{R}^n)$ denote the set of all continuous
$\mathbb{R}^n$-valued functions that are defined on $[- \tau, 0]$. We abbreviate this set as $C_{\rm in}$, and
call it the set of all {\em initial functions}. Also, for any continuous function
$x : [- \tau, \infty) \rightarrow \mathbb{R}^n$ and all $t \geq 0$, we define $x_t$ by
$x_t(\theta) = x(t + \theta)$ for all $\theta \in [- \tau, 0]$, i.e., $x_t \in C_{\rm in}$ is
the translation operator. A vector or
a matrix is nonnegative (resp. positive) if all of its entries are nonnegative (resp. positive).
We write $M \succ 0$ (resp. $M \preceq 0$) to indicate that $M$ is a symmetric positive definite (resp. negative
semi-definite) matrix. For two vectors $V = (v_1 ... v_n)^\top$ and $U = (u_1 ... u_n)^\top$, we write $V \leq U$ to
indicate that for all $i \in \{1,..., n\}$, $v_i \leq u_i$.

\section{LPV Modeling of Decomposable Systems with Switching Topology}
Let us consider an $Nn$-th order interconnected linear time-varying system  
\begin{equation}
\begin{array}{rcl}
\label{equ:systemdef}
\dot{x}(t)&=&{A(t)}x(t)+ {B(t)}u(t)\\ [1mm]
y(t)&=&{C(t)}x(t-\tau(t))
\end{array}
\end{equation}
with  $x\in \mathbb{R}^{Nn}$, $u\in\mathbb{R}^{Nd_u}$, $ y \in \mathbb{R}^{Nd_y}$, $x_t\in C_{\rm in}$, and for all $t \geq 0$, $\tau(t) \in [0, \overline{\tau}]$
with $\overline{\tau} > 0$.

We introduce a range dwell-time condition, i.e., a sequence of real numbers $t_k$ such that there are two positive
constants $\underline{\delta}$ and $\overline{\delta}$ such that $t_0 = 0$ and for all $k \in \mathbb{N}_0$,
\begin{equation}
\label{eq:dist}
t_{k + 1} - t_k \in [\underline{\delta}, \overline{\delta}] \; .
\end{equation} 

We start by formally defining decomposable matrices, which are of interest in describing the systems considered in this paper.
\begin{definition} [\citealp{eichler2014robust}] \label{def:1}
 A  matrix ${M(t)}:\mathbb{R}\rightarrow\mathbb{R}^{Np\times Nq}$ is called \emph{decomposable} if given a matrix $\mathcal{P}(t):\mathbb{R}\rightarrow\mathbb{R}^{N\times N}$ (pattern matrix), there exist matrices $\bar{M}^a$, $\bar{M}^b \in \mathbb{R}^{p\times q}$ such that
\begin{equation} 
\label{eq:def1}
M (t) = I_N\otimes \bar{M}^a + \mathcal{P}(t
)\otimes \bar{M}^b
  \end{equation}
for all $t \ge 0$, where the superscript $a$ represents the \emph{decentralized} part and superscript $b$ represents the \emph{interconnected} part.    $\hfill\square$
\end{definition}
We can now define the class of systems studied in this paper.
\begin{definition}  [\citealp{eichler2014robust}]
The system (\ref{equ:systemdef}) is called \emph{decomposable} if and only if all of its system matrices are decomposable, i.e.,  can
be written as \rref{eq:def1} for the same  matrix  $\mathcal{P}(t)$. $\hfill\square$
\end{definition}

We define the pattern matrix $\mathcal{P}(t)$ as a linear convex combination of two symmetric commutable matrices $\mathcal{P}_1$ and $\mathcal{P}_2$, i.e.,
\begin{equation}
\label{pattern}
\begin{array}{rcl}
\mathcal{P}(t) = \sigma(t) \mathcal{P}_1 + (1- \sigma(t)) \mathcal{P}_2\; ,
\end{array}
\end{equation} 
where $\sigma(t) \in [0,1]$ is a piecewise constant switching signal satisfying the range dwell-time condition \rref{eq:dist}. Moreover, $\dot{\sigma}(t) = 0$ between the jumps and $\sigma(t)$ arbitrarily change its value with a finite jump intensity. Since the symmetric matrices $\mathcal{P}_1$ and $\mathcal{P}_2$ commute with each other, there exists a unitary matrix $U$ that simultaneously diagonalizes the pair $\mathcal{P}_1, \mathcal{P}_2$ according to \cite[Theorem 1.3.12]{horn2012matrix}.   Therefore, we can write \rref{pattern} as
\begin{equation*}
\begin{array}{rcl}
\Lambda(t) = \sigma(t) \Lambda_1 + (1- \sigma(t))\Lambda_2\; ,
\end{array}
\end{equation*}
where $\Lambda(t):\mathbb{R} \mapsto \mathbb{R}^N$ is a matrix-valued function,  $\Lambda_1 \ \text{and}\ \Lambda_2$ are constant block diagonal matrices each of size $N$, and let   $\lambda_{1i}$ and $\lambda_{2i}$ denote the $i$-th eigenvalues of the matrices $\mathcal{P}_1$ and $\mathcal{P}_2$, respectively.

\begin{remark}
The convex combination given in \rref{pattern} is not unique, any linear convex combination is admissible, e.g., 
\begin{equation*}
\begin{array}{rcl}
\mathcal{P}(t) =  {\mathcal{P}_1 + \mathcal{P}_2 +\sigma(t)(\mathcal{P}_2 - \mathcal{P}_1)\over 2 }\; ,
\end{array}
\end{equation*}
where $\sigma(t) \in  [-1,1]$. $\hfill\square$
\end{remark} 

We provide a theorem that will be substantial in proving the results in the sequel.
\begin{theorem}
\label{thm:Theorem 1}
An $Nn$-$th$ order system \rref{equ:systemdef} as described in definition 2 is equivalent to $N$ independent subsystems of order $n$  
\begin{equation}
\begin{array}{rcl}
\label{equ:26}
\dot{\hat{x}}_i(t)&=&A^\dagger(\nu_i)\hat{x}_i(t)+B^\dagger(\nu_i) \hat{u}_i(t) \\ [1mm]
\hat{y}_i(t)&= &C^\dagger(\nu_i) \hat{x}_i(t-\tau(t)) \quad \text{for} \ i = 1,2,...,N\; ,
\end{array}
\end{equation}
where $\hat{x}_i \in \mathbb{R}^{n}$, $\hat{u}_i \in \mathbb{R}^{d_u}$, $\hat{y}_i \in \mathbb{R}^{d_y}$, and  $\nu_i(t) = \sigma(t)\lambda_{1i} + (1 - \sigma(t))\lambda_{2i}$. Moreover, the matrices $A^\dagger(\nu_i)$, $B^\dagger(\nu_i)$, and $C^\dagger(\nu_i)$ are defined as 
\begin{equation}
\label{equ:split}
\begin{array}{rcl}
A^\dagger(\nu_i) &= & {\bar{A}}^a + \nu_i {\bar{A}}^b  \\[1mm]
B^\dagger(\nu_i) &= & {\bar{B}}^a + \nu_i {\bar{B}}^b \\[1mm]
C^\dagger(\nu_i) &=& {\bar{C}}^a + \nu_i {\bar{C}}^b\; .
\end{array}  
\end{equation}
\end{theorem}
\begin{remark}
Interconnected systems in which each subsystem has the same delay, appear in many real-world applications; see  \cite{zhou2014consensus} for the motivation of this assumption. $\hfill\square$
\end{remark}

Let us define the set
\begin{equation*}
\mathscr{P} = 
\left\{
\begin{array}{ccl}
\rho : \mathbb{R}_{\ge 0} \rightarrow P : \rho(t) = p_k \in P, \\[1mm]
t \in [t_k, \ t_{k + 1}), \ k \in \mathbb{N}_0
\end{array}
\right\}\;,
\end{equation*} where $P$ is compact and connected.

We now provide a method to model a decomposable system with switching topology as an LPV system with a piecewise constant parameter.

\begin{theorem}
\label{thm:dectolpv}
The system \rref{equ:systemdef} as described in definition 2 is equivalent to an LPV system  
\begin{equation}
\begin{array}{rcl}
\label{equ:30}
\dot{\omega}(t)&=&\mathcal{A}(\rho)\omega(t)+\mathcal{B}(\rho) v(t) \\ [1mm]
r(t)&= &\mathcal{C}(\rho) \omega(t-\tau(t))\;,
\end{array}
\end{equation}
where $\omega \in \mathbb{R}^{n}$, $v \in \mathbb{R}^{d_u}$, $r \in \mathbb{R}^{d_y}$, and the piecewise constant parameter $\rho\in \mathscr{P}$ satisfies the dwell-time condition \rref{eq:dist} and takes arbitrary values in the interval $[\min\{\underline{\lambda}_1,\underline{\lambda}_2\},\ \max\{\bar{\lambda}_1, \bar{\lambda}_2 \}]$, where $\underline{\lambda}_j,\ \bar{\lambda}_j$ are minimum and maximum eigenvalues of $\mathcal{P}_j$ for $j = 1,2$, respectively. Moreover, the matrices $\mathcal{A}(\rho), \mathcal{B}(\rho),  \text{and}\ \mathcal{C}(\rho)$ are given by
\begin{equation*}
\begin{array}{rcl}
\mathcal{A}(\rho) &= & {\bar{A}}^a + \rho {\bar{A}}^b  \\[1mm]
\mathcal{B}(\rho) &= & {\bar{B}}^a + \rho {\bar{B}}^b \\[1mm]
\mathcal{C}(\rho) &=& {\bar{C}}^a +  \rho{\bar{C}}^b\; .
\end{array}  
\end{equation*}
\end{theorem}

\section{Output Feedback Controller Design}
In this section, we present output feedback stabilizing controller design for the LPV system with a piecewise constant parameter given in \rref{equ:30}.

We start by introducing an assumption.
\begin{assumption}
\label{assump:1}
(i) There exist a matrix $\mathcal{K}(\rho)$  for all $\rho \in \mathscr{P}$ and constants $T \geq \bar{\tau}$, $a \in [0, 1)$, and $b \geq 0$, 
such that the solutions of the system
\begin{equation*}
\label{aux_sys_1}
\dot{\alpha}(t) = \mathcal{M}(\rho) \alpha(t) + \zeta(t)
\end{equation*}
with $\mathcal{M}(\rho) = \mathcal{A}(\rho) + \mathcal{B}(\rho) \mathcal{K}(\rho)$
and $\zeta$ being a piecewise continuous function, satisfy 
\begin{equation*}
\label{eq:assump1_1}
|\alpha(t)| \leq a |\alpha(t - T)| + b \displaystyle\sup_{\ell \in [t - T, t]} |\zeta(\ell)|
\end{equation*}
for all $t \geq T$. \\
(ii) There exist a matrix  $\mathcal{L}(\rho)$ for all $\rho \in \mathscr{P}$ and constants $T \geq \bar{\tau}$,  $c \in [0, 1)$,
and $d \geq 0$, such that the solutions of the system 
\begin{equation*}
\label{aux_sys_2}
\dot{\beta}(t) = \mathcal{N}(\rho) \beta(t) + \eta(t)
\end{equation*}
with
$\mathcal{N}(\rho) = \mathcal{A}(\rho) + \mathcal{L}(\rho) \mathcal{C}(\rho)$
and $\eta$ being a piecewise continuous function, satisfy 
\begin{equation*}
\label{eq:assump1_2}
|\beta(t)| \leq c |\beta(t - T)| + d \displaystyle\sup_{\ell \in [t - T, t]} |\eta(\ell)| 
\end{equation*}
for all $t \geq T$. $\hfill\square$
\end{assumption}

\begin{remark}
Assumption~1 pertains to the stabilizability and the detectability of the system \rref{equ:30}. See Appendix~\ref{ason} below on a method to check Assumption~1. $\hfill\square$
\end{remark}
Let 
\begin{equation}
\begin{array}{l}
\label{eq:theK}
s_1 \triangleq \sup_{\rho \in \mathscr{P}} |\mathcal{B}(\rho) \mathcal{K}(\rho)| \\[1mm]
 s_2 \triangleq \sup_{\rho \in \mathscr{P}} |\mathcal{L}(\rho) \mathcal{C}(\rho)| \\[1mm]
 s_3 \triangleq  \sup_{\rho \in \mathscr{P}} |\mathcal{M}(\rho)|\;,
\end{array}
\end{equation}
then we have the following result:
\begin{theorem}
\label{thm:theorem1}
Let the system (\ref{equ:30}) satisfy Assumption \ref{assump:1}. If for all $t \geq 0$, 
\begin{equation*}
\tau(t) \leq \bar\tau <\bar{\tau}_u,
\end{equation*} 
where
\begin{equation*}
\label{eq:final_result}
\bar{\tau}_u= \frac{(1 - a)(1 - c)}{d s_1 s_2((1 - a) + b s_3)} \;,
\end{equation*}
then the origin of the feedback system 
\begin{equation}
\label{eq:thm1_2}
\left\{
\begin{array}{ccl}
\dot{\omega}(t) & = & \mathcal{A}(\rho) \omega(t) + \mathcal{B}(\rho) \mathcal{K}(\rho) \hat{\omega}(t)
\\[1mm]
\dot{\hat{\omega}}(t) & = & \mathcal{A}(\rho) \hat{\omega}(t) + \mathcal{B}(\rho) \mathcal{K}(\rho) \hat{\omega}(t)
\\[1mm]
& & + \mathcal{L}(\rho) [\mathcal{C}(\rho) \hat{\omega}(t) - r(t)]
\\[1.5mm]
r(t) &=& \mathcal{C}(\rho)\omega(t - \tau(t))
\end{array}
\right.
\end{equation}
is globally uniformly exponentially stable (GUES) for all $\rho  \in  \mathscr{P}$.
\end{theorem}

\begin{remark}
The results of \cite{ahmed2018dynamic} apply only to switched systems with countable modes. Here we extend their results to switched systems with an uncountable number of modes, i.e., parameter-dependent systems with a piecewise constant parameter.
\end{remark}

\begin{proof}
Let us define the error as  $\tilde{\omega}(t) = \hat{\omega}(t) - \omega(t)$.
Then
\begin{equation*}  
\label{eq:thm1_4}
\dot{\tilde{\omega}}(t) = \mathcal{A}(\rho)\tilde{\omega}(t) + \mathcal{L}(\rho)[\mathcal{C}(\rho)\hat{\omega}(t) - \mathcal{C}(\rho) \omega(t - \tau(t))]\; .
\end{equation*}
Using  $\mathcal{M}(\rho) = \mathcal{A}(\rho) + \mathcal{B}(\rho) \mathcal{K}(\rho)$ and $\mathcal{N}(\rho) = \mathcal{A}(\rho) + \mathcal{L}(\rho) \mathcal{C}(\rho)$, we have
\begin{equation*} 
\label{eq:thm1_6}
\left\{
\begin{array}{ccl}
\dot{\omega}(t) & = & \mathcal{M}(\rho) \omega(t) + \mathcal{B}(\rho) \mathcal{K}(\rho)\tilde{\omega}(t)
\\[1mm]
\dot{\tilde{\omega}}(t) & = & \mathcal{N}(\rho) \tilde{\omega}(t) + \mathcal{L}(\rho) \mathcal{C}(\rho) [\omega(t) - \omega(t - \tau(t))]\; .
\end{array}
\right.
\end{equation*}
From Assumption \ref{assump:1} and the equality 
\begin{equation*}
\begin{array}{l}
\omega(\ell) - \omega(\ell - \tau(\ell))
\\
= \int_{\ell - \tau(\ell)}^{\ell} [\mathcal{M}(\rho(m)) \omega(m) + \mathcal{B}(\rho(m)) \mathcal{K}(\rho(m)) \tilde{\omega}(m)] dm,
\end{array}
\end{equation*}
it follows that for all $t \geq T + \bar{\tau}$,
\begin{equation}
\label{eq:thm1_7}
|\omega(t)| \leq a |\omega(t - T)| + b \displaystyle\sup_{\ell \in [t - T, t]} |\mathcal{B}(\rho(\ell)) \mathcal{K}(\rho(\ell)) \tilde{\omega}(\ell)|
\end{equation}
\begin{equation} \hspace{-0.01cm}
\begin{array}{l}
\label{eq:thm1_8}
|\tilde{\omega}(t)|  \leq   c |\tilde{\omega}(t - T)| +
d \displaystyle\sup_{\ell \in [t - T, t]} \left| \vphantom{ \int_{\ell - \tau(\ell)}^{\ell}} \mathcal{L}(\rho(\ell)) \mathcal{C}(\rho(\ell)) \right.
\\
\left. \times \int_{\ell - \tau(\ell)}^{\ell} [\mathcal{M}(\rho(m)) \omega(m) + \mathcal{B}(\rho(m)) \mathcal{K}(\rho(m)) \tilde{\omega}(m)] dm \right| \; .
\end{array}
\end{equation}

Using the constants $s_1$, $s_2$, and $s_3$ defined in \rref{eq:theK}, we deduce from (\ref{eq:thm1_7}) and (\ref{eq:thm1_8}) that $(\omega(t), \tilde{\omega}(t))$ satisfies:
\begin{equation}   \nonumber
\label{eq:thm1_13}
\begin{array}{ccl}
|\omega(t)| & \leq & a |\omega(t - T)| + b s_1\displaystyle\sup_{\ell \in [t - T - \bar{\tau}, t]} |\tilde{\omega}(\ell)|\; ,
\\[4mm]
|\tilde{\omega}(t)| & \leq & d s_2 s_3 \bar\tau \displaystyle\sup_{\ell \in [t - T - \bar{\tau} , t]} |\omega(\ell)|
\\[4mm]
& & + (c + d s_1 s_2 \bar\tau) \displaystyle\sup_{\ell \in [t - T-\bar{\tau} , t]} |\tilde{\omega}(\ell)|\; .
\end{array}
\end{equation}
Lemma~\ref{extlem} from Appendix~A ensures that the origin of (\ref{eq:thm1_2}) is GUES if
\begin{equation} \nonumber
\label{mart}
\left[
\begin{array}{cc}
a & b s_1
\\
d s_2 s_3 \bar\tau  & d s_1 s_2 \bar\tau + c
\end{array}
\right]
\end{equation}
is Schur stable, which is equivalent to
\begin{equation} \nonumber
\label{aort}
\begin{array}{l}
\frac{a + c + ds_1s_2\bar{\tau}}{2} +
\\[3mm]
\sqrt{\left(\frac{a + c + d s_1 s_2 \bar{\tau}}{2}\right)^2 - ac -ds_1s_2 \left(a - bs_3\right)
\bar{\tau}} < 1\; ,
 \end{array}
\end{equation}
from which we derive the  condition 
\begin{equation}
\bar{\tau}< \frac{(1 - a)(1 - c)}{d s_1 s_2((1 - a) + b s_3)} \;.
\end{equation}
This concludes the proof. $\hfill\blacksquare$
\end{proof}  

\section{Application to Multi-agent Nonholonomic Systems}
In this section, we illustrate our approach by applying it to the consensus problem of  multi-agent nonholonomic systems subject to switching topology and communication delay. The dynamics of the multi-agent nonholonomic system  is adopted from \cite{gonzalez2014lpv}.  

Consider the multi-agent system comprising of six agents ($N=6$) described by
\begin{equation}
\label{example}
\begin{array}{rcl}
\dot{x}(t) &=& (I_N \otimes \bar{A}_a + \mathcal{P}(t) \otimes \bar{A}_b)x(t) 
\\[1mm]
&&+ (I_N \otimes \bar{B}_a)u(t) \\ [1mm] 
y(t) &=& (I_N \otimes \bar{C}_a)x(t - \tau(t))    \; ,
\end{array}
\end{equation} 
 where the system matrices are given by
\begin{equation*}
\small
\begin{array}{lll}
\bar{A}_a = 
\left[
\begin{matrix}
0 & 1 \\ -1 & 0 
\end{matrix}
\right], \
\bar{A}_b = 
\left[
\begin{matrix}
0 & -0.5 \\ 0.5 & 0 
\end{matrix}
\right] \\ [3mm]
\bar{B}_a = 
\left[
\begin{matrix}
1 & 0 \\ 0 & 0.6
\end{matrix}
\right], \ 
\bar{C}_a = 
\left[ 
\begin{matrix}
1 & 0 \\ 0 & 1
\end{matrix}
\right].
\end{array}
\end{equation*} 
For the multiagent system  \rref{example}, the  pattern matrix is specified as 
\begin{equation*}
\begin{array}{rcl}
\mathcal{P}(t) = \sigma(t)\mathcal{P}_1 + (1 - \sigma(t))\mathcal{P}_2\; ,
\end{array}
\end{equation*}
where $\sigma(t) \in [0,\ 1]$ with $\underline{\delta} = 0.1, \ \bar{\delta} = 0.5 $, and the  symmetric commutable matrices $\mathcal{P}_1$ and $\mathcal{P}_2$ are given by
\begin{equation} \nonumber
\mathcal{P}_1 = 
\left[
\begin{matrix}
1 &-0.5 & 0& 0& 0 &-0.5 \\ 
    -0.5 & 1& -0.5& 0 &0& 0\\
    0 &-0.5& 1 &-0.5& 0 &0 \\
    0 &0 &-0.5 & 1 &-0.5& 0  \\ 
    0 &0 &0 &-0.5& 1 &-0.5 \\
    -0.5& 0& 0 &0 &-0.5 &1  \\
\end{matrix}
\right]
\end{equation}
\begin{equation} \nonumber
\mathcal{P}_2 =
\left[
\begin{matrix}
1 & -0.25 &-0.25& 0& -0.25 & -0.25  \\
    -0.25& 1 &-0.25 &-0.25 & 0 & -0.25 \\ 
    -0.25 &-0.25& 1 &-0.25 & -0.25 & 0   \\
    0 &-0.25 &-0.25& 1 & -0.25 & -0.25 \\
    -0.25& 0& -0.25 &-0.25 & 1  & -0.25   \\   
    -0.25 &-0.25& 0 &-0.25 &-0.25& 1 \\
\end{matrix}
\right]\; .
\end{equation}
For the matrices $\mathcal{P}_1$ and $\mathcal{P}_2$, we have $\min\{\underline{\lambda}_1,\ \underline{\lambda}_2 \}= 0$, $\max\{\bar{\lambda}_1,\ \bar{\lambda
}_2 \} = 2$.  Therefore, by defining a piecewise constant parameter $\rho \in \mathscr{P} \ \text{that takes values in} \  [ 0, 2 ]$, and then employing Theorem~2, the decomposable system \rref{example} is equivalent to the  LPV system 
\begin{equation*}
\begin{array}{rcl}
\dot{\omega}(t)&=&\mathcal{A}(\rho)\omega(t)+\mathcal{B}(\rho) v(t) \\ [1mm]
r(t)&= &\mathcal{C}(\rho) \omega(t-\tau(t))\; ,
\end{array}
\end{equation*}
where 
\begin{equation*}
\begin{array}{rcl}
\mathcal{A}(\rho) = \bar{A}_a + \rho \bar{A}_b,\ \mathcal{B}(\rho) = \bar{B}_a, \ \text{and}\
\mathcal{C}(\rho) = \bar{C}_a\;.
\end{array}
\end{equation*}    

We choose the controller gains and observer gains as
\begin{equation*}
\small
\begin{array}{rcl}
\bar{K}_a &=& \bar{L}_a = \vspace{1mm}
\left[
\begin{matrix}
-0.5 & 0 \\ 0  & -0.5 
\end{matrix}
\right] ,  \\
\bar{K}_b &=& \bar{L}_b =
\left[ 
\begin{matrix}
0.1 & 0 \\ 0 & 0.1
\end{matrix}
\right].
\end{array}
\end{equation*}
Using Theorem~2, we model the distributed controller $K(t) = I_N \otimes \bar{K}_a + \mathcal{P}(t) \otimes \bar{K}_b$ and the distributed observer $L(t) = I_N \otimes \bar{L}_a + \mathcal{P}(t) \otimes \bar{L}_b$  as $\mathcal{K}(\rho)= \bar{K}_a + \rho \bar{K}_b$ and $\mathcal{L}(\rho) = \bar{L}_a + \rho \bar{L}_b$, respectively. 

In order to satisfy Assumption~1, we proceed as follows. First, we solve the LMIs \rref{eq:prelim_20}, \rref{eq:prelim_21}, and \rref{eq:prelim_22} in Lemma~3~(Appendix~B) by setting $\Omega(\rho)=\mathcal{M}(\rho)=\mathcal{A}(\rho)+\mathcal{B}(\rho)\mathcal{K}(\rho)$. This yields $d_1=0.01$, $d_2=0.01$, $\mu=1$, $\gamma=1$, $\eta=50$, and $T=5$. Therefore, part (i) of Assumption~1 is satisfied with 
\begin{equation*}
\begin{array}{rcl}
a&=& \sqrt{\frac{d_2}{d_1} \mu  (\mu e^{-\gamma \underline{\delta}})^{\eta}e^{\gamma \overline{\delta}}} = 0.1054\\
 b &=& \sqrt{\mu \frac{d_2}{\gamma d_1} T}= 2.2361.
 \end{array}
\end{equation*} 
Then, we set $\Omega(\rho)=\mathcal{N}(\rho)=\mathcal{A}(\rho)+\mathcal{L}(\rho)\mathcal{C}(\rho)$ in Lemma~3 (Appendix~B), and again solve the LMIs \rref{eq:prelim_20}, \rref{eq:prelim_21}, and \rref{eq:prelim_22} yielding $d_1=0.01$, $d_2=0.01$, $\mu=1$, $\gamma=1$, $\eta=75$, and $T=7.5$. Therefore, part (ii) of Assumption~1  is satisfied with  
\begin{equation*}
\begin{array}{rcl}
c &=& \sqrt{\frac{d_2}{d_1} \mu  (\mu e^{-\gamma \underline{\delta}})^{\eta}e^{\gamma \overline{\delta}}}= 0.0302 \\
d &=&\sqrt{\mu \frac{d_2}{\gamma d_1} T}=2.7386.
 \end{array}
\end{equation*} 
Moreover, $s_1 = 0.5$ , $s_2= 0.5$, $s_3 = 1.177 $. According to Theorem \ref{thm:theorem1}, the closed-loop system is GUES for $\tau(t)\leq \bar{\tau}<0.3593$.

\subsection{Computational Aspects}
The LMIs \rref{eq:prelim_20}, \rref{eq:prelim_21}, and \rref{eq:prelim_22} obtained in Lemma~3~(Appendix~B) take the form of infinite-dimensional semidefinite program. In order to check their feasibility, we propose gridding method. The idea is to approximate semi-infinite constraint LMI by a finite number of of LMIs, \cite{briat2014linear}, that can be implemented using YALMIP, \cite{lofberg2004yalmip}, and solved using semidefinite programming solver such as SeDuMi, \cite{sturm1999using}.

\subsection{Simulation Results}
 The  closed-loop is simulated subject to time-varying delay $\tau(t) = 0.05 sin(4t) + 0.3$. Fig. 1 shows the evolution of the state trajectories and the switching signal for the simulation setup. It is evident that the state trajectories reach a consensus. For the switching signal shown in Fig.~1, the time between two consecutive jumps is a uniform random variable that takes values in the compact set $[\underline{\delta}, \ \bar{\delta}]$, and the value of $\sigma(t)$ at each jump is also uniform random variable that takes values in the compact set  $[0, \ 1]$.   The consensus of multiagent system in Fig.~1 subject to  switching topology and time varying delay  reflects the efficacy of the approach. 
 
\begin{figure}
\includegraphics[width = 0.9\columnwidth]{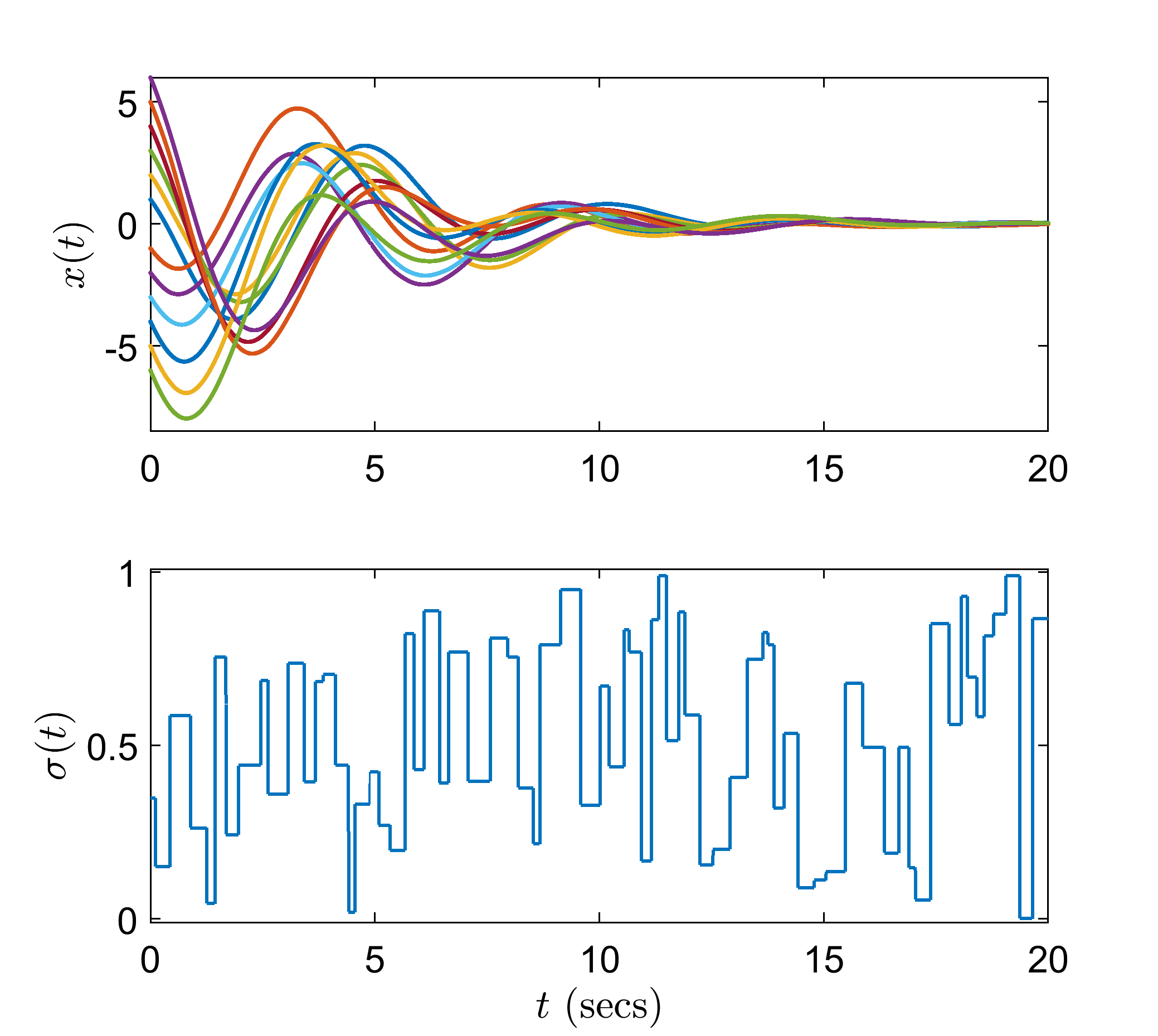}
\label{fig:states}
\caption{State trajectories (top),  switching signal  (bottom).}
\end{figure}


\section{Conclusion}
It has been shown in this paper that decomposable systems with switching topology can be modeled as LPV systems with a piecewise constant parameter. Then, output feedback stabilizing controllers are designed for these LPV systems in the presence of a time-varying output delay.  A trajectory-based approach is used for stability analysis which circumvents
the serious obstacle presented by the search for appropriate Lyapunov functionals. No constraint is imposed on the delay derivative. Output feedback control of distributed systems subject to stochastic topology following the idea \cite{zakwan2019poisson} seems to be one of the promising future extensions of the current work. Moreover, employing different kinds of the controller (e.g., memory or memory resilient) and observer structures (e.g., interval observers) can be quite intriguing. These extensions will be reported elsewhere.

\begin{ack}
The authors would like to acknowledge useful discussions with Professor Hitay \"Ozbay and Dr. Corentin Briat.
\end{ack}

\bibliography{ifacconf}  

\appendix
\section{Technical Lemmas}    
In this section, we provide technical lemmas. Lemma~1 highlights an interesting property of decomposable matrices and it is used to prove Theorem~\ref{thm:Theorem 1}. Lemma~2 recalls the trajectory based stability analysis approach from \cite{ahmed2018dynamic} and it is used to prove Theorem~\ref{thm:theorem1}. 
\begin{lemma}
\label{eq:lem1}
Consider a  matrix $M(t)$ with the structure \rref{eq:def1} subject to the pattern matrix $\mathcal{P}(t)$ given in \rref{pattern}, then the matrix 
\[
M^\dagger (t)=(U\otimes I_p)^{-1}M(t)(U\otimes I_q)
\]
is block diagonal and has the following structure
\begin{equation}
\label{M_dagger}
{M}^\dagger(t) = I_N\otimes \bar{M}^a + \Lambda(t) \otimes \bar{M}^b\; ,
\end{equation}
where each of the block has the form ${M}^ \dagger_i(t)=\bar{M}^a + \nu _i(t) \bar{M}^b $ where $\nu_i(t) = \sigma(t)\lambda_{1i} + (1 - \sigma(t))\lambda_{2i}$. Moreover, for every  matrix $M^\dagger(t)$ with the structure \rref{M_dagger}, we have 
\[
M(t) =(U\otimes I_p)M^\dagger(t)(U\otimes I_q)^{-1}
=I_N\otimes \bar{M}^a + \mathcal{P}(t)\otimes \bar{M}^b.
\] 
\end{lemma}
\begin{proof}
From Definition 1, we can write 
\begin{equation*}
{M}^\dagger(t) = (U \otimes I_p )^{-1} (I_N \otimes \bar{M}^a + \mathcal{P}(t) \otimes \bar{M}^b) (U \otimes I_q )
\end{equation*}
then from the properties of the Kronecker product \citep{brewer1978kronecker}, we have 
\begin{equation*}
\begin{array}{lll}
{M}^\dagger(t) = (U^{-1} I_N U \otimes I_p \bar{M}^a I_q ) +  ( U^{-1} \mathcal{P}(t) U \otimes I_p \bar{M}^b I_q )\;. 
\end{array}
\end{equation*}
As an immediate consequence,
\begin{equation*}
\begin{array}{lll} 
{M}^\dagger(t) = I_N \otimes \bar{M}^a + \Lambda(t) \otimes \bar{M}^b \; .
\end{array}
\end{equation*}
Since $I_N$ and $\Lambda(t)$ are diagonal, therefore,  ${M}^{\dagger}(t)$ is block diagonal. The converse can be proved analogously. $\hfill\blacksquare$
\end{proof}

\begin{lemma} [\citealp{ahmed2018dynamic}]
\label{extlem}
Let us consider a constant $T > 0$ and $l$ functions $z_g : [- T, + \infty) \rightarrow [0, + \infty)$, $g = 1, ... , l$.
Let $Z(t) = (z_1(t) \; ... \; z_l(t))^\top$ and, for any $\theta \geq 0$ and $t \geq \theta$, define
$\mathfrak{V}_{\theta}(t) = \left(\displaystyle\sup_{s \in [t - \theta, t]} z_1(s) \; ...
\displaystyle\sup_{s \in [t - \theta, t]} z_l(s)\right)^\top$.
Let $\Upsilon \in \mathbb{R}^{l \times l}$ be a nonnegative Schur stable matrix.
If for all $t \geq 0$, the inequalities
$Z(t) \leq \Upsilon \mathfrak{V}_T(t)$
are satisfied, then
$\displaystyle\lim_{t \rightarrow + \infty} z_g(t) = 0
~~~\forall~  g = 1,\hdots, l$. 
\end{lemma}

\section{Checking Assumption 1}
\label{ason}

In this section, we illustrate a method to determine the constants $a$, $b$, $c$, and $d$ to satisfy  Assumption \ref{assump:1}. \\
Consider an LPV system subject to piecewise parameter trajectory
\begin{equation}
\label{eq:prelim_1}
\dot{\xi}(t) = \Omega(\rho) \xi(t) + \vartheta(t)\; ,
\end{equation}
where $\xi \in \mathbb{R}^{d_{\xi}}$, $\rho \in \mathscr{P}$ and $\vartheta$ is a piecewise continuous function.
\begin{lemma}
\label{lem:lemma2}
Let the system (\ref{eq:prelim_1}) be such that there are real numbers $d_1 > 0$, $d_2 > 0$, $\mu \geq 1$, $\gamma > 0$
and symmetric positive definite matrices $\mathcal{Q}(\rho)$, such that the LMIs
\begin{eqnarray}
~~~~&~~&
d_1 I \preceq \mathcal{Q}(\rho) \preceq d_2 I\; ,\label{eq:prelim_20} \\
~~~~&~~&
\mathcal{Q}(\rho) \preceq \mu \mathcal{Q}(\theta)\; ,\label{eq:prelim_21}\\
~~~~&~~&
\Omega(\rho)^\top \mathcal{Q}(\rho) + \mathcal{Q}(\rho) \Omega(\rho) \preceq  - \gamma \mathcal{Q}(\rho) \label{eq:prelim_22}
\end{eqnarray}
are satisfied for all $\rho, \ \theta \in \mathscr{P}$.
Moreover, the constant $\mu_{\triangle} = \mu e^{-\gamma \underline{\delta}}$ is such that
\begin{equation*}
\label{eq:prelim_5}
\mu_{\triangle} < 1\; .
\end{equation*}
Then, along the trajectory of (\ref{eq:prelim_1}), the inequality
\begin{equation} \nonumber
\label{eq:prelim_6}
|\xi (t)| \leq  \sqrt{\frac{d_2}{d_1} \mu  \mu_{\triangle}^{\eta}e^{\gamma \overline{\delta}}} |\xi(t - T)|
+ \sqrt{\mu \frac{d_2}{\gamma d_1} T} \displaystyle\sup_{\ell \in [t - T, t]} |\vartheta(\ell)|
\end{equation}
holds for all $t \geq T$ where $T>0$ and $\eta$ is a positive integer depending on the choice of $T$ such that for all $t\in[t_k, t_{k+1})$, we have $t-T\in[t_{k-\eta-1}, t_{k-\eta})$. Moreover, we have  $\sqrt{\frac{d_2}{d_1} \mu  \mu_{\triangle}^{\eta}e^{\gamma \overline{\delta}}}<1$ when $\eta > \frac{1}{\ln(\mu_{\Delta})}
\left[\ln \left(\frac{d_1}{d_2\mu}\right) - \gamma \overline{\delta}\right]$. 
\end{lemma}

\section{Proofs of the Theorems} 
In this section, we provide proofs of the theorems appearing in Section~2.
\begin{proof} [Proof of Threorem~\ref{thm:Theorem 1}:]
Using Lemma \ref{eq:lem1}, we can rewrite \rref{equ:systemdef} as
\begin{equation}
\label{equ:8be}
\begin{array}{lll}
(U \otimes I_n)^{-1} \dot{x}(t)&=&A^\dagger(t)(U \otimes I_n)^{-1} x(t)\\[1mm]
&&+B^\dagger(t) (U\otimes I_{d_u})^{-1} u(t)  \\[1mm]
(U \otimes I_{d_y})^{-1} y(t)&=& C^\dagger(t) (U \otimes I_n)^{-1} x(t) \; .
\end{array}
\end{equation}
Then, with  change of variables $x(t) = (U \otimes I_n)\hat{x}(t)$, $u(t) = (U\otimes I_{d_u})\hat{u}(t)$, and $y(t) = (U \otimes I_{d_y})\hat{y}(t)$, it follows from \rref{equ:8be} that
\begin{equation}
\begin{array}{rcl}
\label{equ:10be}
\dot{\hat{x}}(t)&=& A^\dagger(t)\hat{x}(t)+B^\dagger(t)\hat{u}(t) \\ [1mm]
\hat{y}(t)&= & {C}^\dagger(t)\hat{x}(t)\; ,
\end{array}
\end{equation}
where 
\begin{equation*}
\begin{array}{lll}
{A}^\dagger(t) = I_N \otimes \bar{A}^a + \Lambda(t) \otimes \bar{A}^b \\
{B}^\dagger(t) =I_N \otimes \bar{B}^a + \Lambda(t) \otimes \bar{B}^b \\[1mm]
C^\dagger(t) = I_N \otimes \bar{C}^a + \Lambda(t) \otimes \bar{C}^b \\[1mm]
\end{array}
\end{equation*}  
are block diagonal matrices. Therefore, the system \rref{equ:10be} is equivalent to $N$ independent $nth$ order subsystems given in \rref{equ:split}. This concludes the proof.  $\hfill\blacksquare$
\end{proof}         

\begin{proof}[Proof of Theorem~\ref{thm:dectolpv}:]
The proof is straightforward and relies on the  affine dependence of system matrices on $\nu_i$ in \rref{equ:26}. Each $\nu_i \in [\lambda_{1i}, \ \lambda_{2i}]$ can be substituted with a bigger polytope $\rho \in [\min\{\nu_i\}, \ \max\{\nu_i\}]$ for $i = 1,2,\hdots,N$. It is obvious that $\min\{\nu_i\}= \min\{\underline{\lambda}_1, \ \underline{\lambda}_2\}$. Arguing similarly, we have  $\max\{\nu_i\} = \max\{\bar{\lambda}_1, \ \bar{\lambda}_2\}$. With these substitutions, the dependence of system matrices on $\nu_i$ can be dropped and system of $N$ independent  subsystems can be represented by an equivalent LPV system \rref{equ:30}. Since the systems \rref{equ:systemdef} and \rref{equ:26} are equivalent according to Theorem~\ref{thm:Theorem 1}, the LPV framework \rref{equ:30} captures the dynamics of the system \rref{equ:systemdef}.  This completes the proof.   $\hfill\blacksquare$ 
\end{proof} 
\end{document}